\ifdefined\pdfminorversion\pdfminorversion=7\fi
\documentclass[runningheads]{llncs}
\usepackage[T1]{fontenc}
\usepackage[pdftex]{graphicx}
\usepackage{xcolor}
\usepackage{color}
\usepackage{comment}
\usepackage{tablefootnote}
\usepackage[hang,small,bf]{caption}
\usepackage[subrefformat=parens]{subcaption}
\usepackage[colorlinks=true,linkcolor=magenta,citecolor=blue]{hyperref}
\usepackage{booktabs}
\usepackage{multirow}
\usepackage{amsmath}
\usepackage{amssymb}

\newcommand{\dist}{\mathsf{dist}}

\begin{document}
\title{Constant Factor Optimal 2-Resilient Local Failover Routing Scheme on Directed Graphs}
\titlerunning{Constant Factor Optimal 2-Resilient Local Failover Routing Scheme}
%
\author{Kaito Harada \and Naoki Kitamura \and Yuki Kawashima \and Yuya Terashima}
\authorrunning{K. Harada et al.}
%
\institute{The University of Osaka, Osaka, Japan\\
\email{\{k-harada,n-kitamura,y-kawashima,t-yuya\}@ist.osaka-u.ac.jp}}
\maketitle              
\begin{abstract}
  Local failover routing delivers a packet from a source $s$ to a destination $t$ using only pre-computed, per-node forwarding rules together with a small rewritable packet header, remaining correct even when up to $k$ arcs of the network fail.
  We study this problem on directed graphs with $n$ nodes, measuring efficiency by the number of rewritable bits carried in the header, and focus on the case of $k \ge 2$ failures.
  Our first contribution is an improved upper bound for general $k$.
  We design a scheme using at most $\lceil \log\!\binom{2n+k-4}{k-1} \rceil+1$ bits, improving the best previously known $\lceil \log \binom{2n+k-3}{k} \rceil$ bit scheme for every $k < 2n-3$.
  For the special case $k=2$, we further improve this to a scheme using at most $\left\lceil \log\!\left(1+2\lfloor 2(n-1)/3 \rfloor\right)\right\rceil$ bits.
  Our second contribution is a family of improved lower bounds, obtained via a new path-gadget construction.
  For $k=2$, we prove a lower bound of $\left\lceil \log \lfloor n/3 \rfloor \right\rceil$ bits, and for general $k$, writing $k_2 = \lfloor k/2 \rfloor$, a lower bound of $\left\lceil k_2 \log \left\lfloor (n-1+k_2)/(3k_2) \right\rfloor \right\rceil$ bits, improving on the best previously known bound.
  In particular, for $k=2$, our upper bound of $\left\lceil \log\!\left(1+2\lfloor 2(n-1)/3 \rfloor\right)\right\rceil$ bits and our lower bound of $\left\lceil \log \lfloor n/3 \rfloor \right\rceil$ bits differ by at most $3$ bits.

  \keywords{failover routing \and networks \and directed graphs.}
\end{abstract}
\section{Introduction}
\subsection{Background and Our Results}
In communication networks~\cite{CKRRS21}, maintaining routing connectivity is critically important, and hence modern network systems implement some form of fast failover routing to respond quickly to equipment failures~\cite{RH20}. This is a mechanism that forwards packets from the source to the destination using only pre-computed routing tables, even when multiple links have failed.
We say a local failover routing scheme is \emph{$k$-resilient} if, when up to $k$ link failures occur, packets are correctly delivered between any pair of nodes that remain connected after the failures. In this process, routing uses only pre-computed local forwarding rules and rewritable packet headers, without any network-wide recalculation.

There are many studies of failover routing on undirected graphs, and it has been investigated whether routing is possible depending on the number of bits in the packet header, the type of failure, the graph properties, etc. In contrast, there is not much research on directed graphs. Several previous studies have established the limits of the number of bits required in packet headers for local failover routing in directed graphs. However, an optimal routing scheme has been shown only for the case of a single failure, and a gap between the upper and lower bounds still remains when multiple failures occur. In this paper, we present a nearly optimal routing scheme for cases in which the number of failures lies within a certain range.

The local failover routing for directed networks is just at the beginning.
Grobe et al.~carefully examined whether existing methods for undirected graphs can be applied to directed graphs, and compared the existing methods with their proposed method through simulations. As a result, they showed that their application is superior to the existing algorithms in most topologies~\cite{grobe2024local}.
van den Akker et al.~analyzed real-world topologies and showed that two-edge failure tolerance is possible in many real networks using two additional bits~\cite{van2024short}. However, these two earlier results do not provide any theoretical analyses.
van den Akker et al.~investigated local failover routing in directed graphs and gave upper and lower bounds on the number of bits required in the packet header~\cite{EK24}. They showed that at least one bit is required for one failure in the packet header and that $\log n$ bits are sufficient to tolerate one failure. They also showed that at least $\lceil \log (k+1) \rceil$ bits are required for $k > 1$ failures, and $k \lceil\log|E|\rceil$ bits are sufficient to tolerate $k$ failures.

Subsequently, several improvements to these upper and lower bounds were presented.
van den Akker et al.~showed that routing is possible using a single bit when there is one failure and proposed an optimal method for handling a single failure~\cite{VWF26}. They also improved the existing lower bound for cases involving two or more failures~\cite{VWF}. Furthermore, Kawashima et al.~proposed a new routing method and lower-bound graph, thereby improving the upper and lower bounds for scenarios with two or more failed arcs~\cite{KKI26}. As a result, they demonstrated a routing method that is asymptotically near-optimal when the number of failures falls within a specific range. Table~\ref{tab:previous} summarizes the results of these existing studies.

Thus, in previous work on local failover routing in directed graphs, no routing scheme achieving an optimal number of bits has been proposed for the case where the number of edge failures is $2$, or for the case where the number of failures is roughly greater than the number of vertices. In this paper, we present new results for these cases.

\paragraph{Contribution.}
In this paper, we present two new schemes and improve the lower bound for general $k \ge 2$.

Our first contribution is an improved upper bound for a $k$-resilient local failover routing scheme for any $2 \le k < 2n-3$.
Our proposed scheme $\mathcal{C}_k$ uses at most
$\lceil \log\binom{2n+k-4}{k-1}\rceil+1$ bits
(Theorem~\ref{thm:multi-failure}).
For every $k < 2n-3$, it improves the underlying state-space bound of
the previous scheme~\cite{KKI26}; the resulting integer bit bounds may
coincide for some values of $n$ and $k$ because of rounding.

Our second contribution is a further-improved 2-resilient scheme.
For the special case $k=2$, restricting this construction to a single well-chosen $s$--$t$ path yields a further-improved scheme $\mathcal{D}$ using at most $\left\lceil \log\!\left(1+2\lfloor 2(n-1)/3 \rfloor\right)\right\rceil$ bits (Theorem~\ref{thm:2-resilient-path}).

Our third contribution is a family of improved lower bounds (Theorem~\ref{thm:lb-general-k}), obtained via a new path-gadget construction that forces the packet header to distinguish among many candidate failure patterns: a lower bound of $\left\lceil \log \lfloor n/3 \rfloor \right\rceil$ bits for $k=2$, and, by cascading $\lfloor k/2 \rfloor$ copies of the gadget, a lower bound of $\left\lceil \lfloor k/2 \rfloor \log \left\lfloor (n-1+\lfloor k/2\rfloor)/(3\lfloor k/2\rfloor) \right\rfloor \right\rceil$ bits for general $k$, both improving on the best previously known bounds.

Among all previously known results (Table~\ref{tab:previous}), $k=2$ was the sole exception where the upper and lower bounds were not asymptotically tight: for $k=1$ the bound was exactly $1$ bit, and for $k \ge 3$ the upper and lower bounds were both $\Theta(k\log(n/k))$ bits~\cite{KKI26}. For $k=2$, however, the best known lower bound was only $\left\lceil \log \log (\lfloor n/4 \rfloor-2) \right\rceil+2$ bits~\cite{VWF} --- doubly logarithmic in $n$ --- against a $\Theta(\log n)$-bit upper bound~\cite{KKI26}, so the gap grew unboundedly with $n$. Our results close this last exception: for $k=2$, our new upper and lower bounds narrow the gap to within $3$ bits, and, more generally, whenever $k=\Theta(n)$ within the range where our lower bound applies (i.e., $\lfloor k/2 \rfloor \le (n-1)/8$), both this lower bound and the upper bound of $\mathcal{C}_k$ are $\Theta(n)$ bits --- so the upper and lower bounds are now asymptotically tight across the entire range $k=O(n)$.
Table~\ref{tab:related} summarizes the known results and our new results on local failover routing.

\begin{table}[t]
  \centering
  \caption{Summary of previous works, where $k_i = \lfloor k/i \rfloor$.}
  \label{tab:previous}
  \setlength{\tabcolsep}{6pt}
  \renewcommand{\arraystretch}{1.3}
  \resizebox{\linewidth}{!}{%
    \begin{tabular}{c|c|c}
      Failures         & Upper Bound                                                                          & Lower Bound                                                                          \\
      \hline
      $k=1$            & $1$ bit~\cite{VWF26}                                                                 & $1$ bit~\cite{EK24}                                                                  \\
      \hline
      $k=2$            & $\left\lceil \log \binom{2n-1}{2} \right\rceil$ bits~\cite{KKI26}                    & $\left\lceil \log \log (\lfloor n/4 \rfloor-2) \right\rceil+2$~\cite{VWF}            \\
      \hline
      $k=3$            & \multirow{3}{*}{$\left\lceil \log\!\binom{2n+k-3}{k}\right\rceil$ bits~\cite{KKI26}} & $\left \lceil \log (n-1) \right \rceil-1$ bits~\cite{KKI26}                          \\
      \cline{1-1}\cline{3-3}
      $k \le 3(n-1)/8$ &                                                                                      & $\left\lceil k_3 \log \lfloor (n-1)/(2k_3)  \rfloor  \right\rceil$ bits~\cite{KKI26} \\
      \cline{1-1}\cline{3-3}
      $3(n-1)/8<k$     &                                                                                      & $\left \lceil (n-1)/4 \right \rceil$ bits~\cite{KKI26}                               \\
      \hline
    \end{tabular}%
  }
\end{table}

\begin{table}[t]
  \centering
  \caption{Summary of the best known results after incorporating our contributions, where $k_i = \lfloor k/i \rfloor$.}\label{tab:related}
  \setlength{\tabcolsep}{6pt}
  \renewcommand{\arraystretch}{1.3}
  \newcommand{\Stack}[2]{\begin{tabular}[t]{@{}c@{}}\rule{0pt}{2.4ex}#1\\[2pt]#2\rule[-1ex]{0pt}{0pt}\end{tabular}}
  \resizebox{\linewidth}{!}{%
    \begin{tabular}{c|c|c}
      Failures         & Upper Bound                                                                                                               & Lower Bound                                                                                                                    \\
      \hline
      $k=1$            & $1$ bit~\cite{VWF26}                                                                                                      & $1$ bit~\cite{EK24}                                                                                                            \\
      \hline
      $k=2$            & \Stack{$\left\lceil \log\!\left(1+2\lfloor2(n-1)/3\rfloor\right)\right\rceil$ bits}{(Theorem~\ref{thm:2-resilient-path})} & \Stack{$ \left \lceil \log \lfloor n/3  \rfloor \right \rceil$ bits}{(Theorem~\ref{thm:lb-general-k})}                         \\
      \hline
      $k=3$            & \multirow{4}{*}{\Stack{$\left\lceil \log\!\binom{2n+k-4}{k-1}\right\rceil + 1$ bits}{(Theorem~\ref{thm:multi-failure})}}  & $\left \lceil \log (n-1) \right \rceil-1$ bits~\cite{KKI26}                                                                    \\
      \cline{1-1}\cline{3-3}
      $k \le (n-1)/4$  &                                                                                                                           & \Stack{$\left\lceil k_2 \log \left\lfloor (n-1+k_2)/(3k_2) \right\rfloor \right\rceil$ bits}{(Theorem~\ref{thm:lb-general-k})} \\
      \cline{1-1}\cline{3-3}
      $k \le 3(n-1)/8$ &                                                                                                                           & $\left\lceil k_3 \log \lfloor (n-1)/(2k_3)  \rfloor  \right\rceil$ bits~\cite{KKI26}                                           \\
      \cline{1-1}\cline{3-3}
      $3(n-1)/8<k$     &                                                                                                                           & $\left \lceil (n-1)/4 \right \rceil$ bits~\cite{KKI26}                                                                         \\
      \hline
    \end{tabular}%
  }
\end{table}

\subsection{Related Work}
There is a lot of research on failover routing in undirected networks.
Feigenbaum et al.~first presented a theoretical study on local failover routing~\cite{FGPSSS12}. They showed that, without header information, single-edge failures are always tolerable, but tolerating multiple failures is generally impossible.
Dai et al.~investigated the limitation by Feigenbaum et al.~in more detail. Precisely, they showed that the routing without rewriting packet headers is possible for two-edge failures, but impossible for three or more failures~\cite{dai2023tight}.
Chiesa et al.~proposed a routing method that can tolerate $k-1$ failures using $\log k$ bits in $k$-edge-connected graphs for any $k \leq 5$~\cite{chiesa2016resiliency}.
Foerster et al.~studied fault tolerance in two variants of the models in which vertices can and cannot identify the source of a packet. They showed that perfect fault tolerance is impossible for nonplanar graphs. They also proposed algorithms achieving perfect fault tolerance for all outerplanar graphs and related settings, as well as for nonouterplanar graphs where the destination is within two hops of the source~\cite{foerster2021feasibility}.
Dai et al.~classified link failures into three types: static (i.e., links which permanently and simultaneously fail), semi-dynamic (removing the assumption that links fail simultaneously), and dynamic (removing the assumption that links fail permanently), and examined the fault tolerance of failover routing for each type, clarifying its capabilities and limitations~\cite{dai2024dynamic}. As a result, they showed a routing method which tolerates $k-1$ dynamic link failures in $k$-edge-connected graphs for $k \leq 5$. Furthermore, they showed that this result can be extended to any $k$ by providing $\log k$ rewritable bits in the packet header. They also showed that rewriting $3$ bits suffices to cope with $k$ semi-dynamic failures. However, on general graphs, tolerating $2$ dynamic failures becomes impossible without rewritable bits. Even by rewriting $\log k$ bits, fault tolerance is impossible for $k$ dynamic failures.

\section{Preliminaries}
\subsection{Notation}

We work with a simple directed graph $G = (V, E), |V| = n$, where nodes represent routers and arcs represent communication links. Each node is assigned a unique identifier. Each arc is also assigned a unique identifier represented by the pair of its endpoints' IDs.
An arc from $u$ to $v$ is written $(u, v)$; $u$ is its \emph{tail} and $v$ its \emph{head}.
For a set $F \subseteq E$ of arcs, $G \setminus F$ denotes the subgraph $(V, E \setminus F)$; we abbreviate $G \setminus \{e\}$ as $G \setminus e$.
A \emph{directed path} from $u$ to $v$ is a sequence of vertices $(u = w_0, w_1, \ldots, w_\ell = v)$ such that $(w_i, w_{i+1}) \in E$ for each $0 \le i < \ell$.
Note that this path is not necessarily simple.
Vertex $v$ is \emph{reachable} from $u$ in $G$ if such a path exists.
An \emph{$s$--$t$ path} is a directed path from source $s$ to destination $t$.

For a (weighted) directed graph $H$ and vertices $u, v$, we write $\dist_H(u, v)$ for the shortest distance from $u$ to $v$ in $H$; we only use this notation for weighted graphs, in which the length of a path is the sum of the weights of its arcs.

For a vertex set $S \subseteq V$, let $\delta^+(S) = \{(u,v) \in E \mid u \in S, v \notin S\}$ denote the \emph{outgoing cut} of $S$.
If $S$ is a singleton $\{v\}$, we omit the braces and write $\delta^+(v)$.
An \emph{in-arborescence} rooted at $r$ is a spanning subgraph in which every vertex $v \ne r$ has exactly one outgoing arc and the unique path following these arcs from any $v$ leads to $r$.
Two arborescences are \emph{arc-disjoint} if they share no common arc.

\subsection{Model}
Given a directed graph $G=(V,E)$ and a failure parameter $k$, a local failover routing $R(G,k)$ consists of one \emph{local forwarding rule} for each vertex in $V$.
A local forwarding rule $R_v$ at $v$ is defined as the following function:
$$
  R_v:E \times 2^{|\delta^+(v)|} \times V \times V \times \{0,1\}^{\ast} \to E \times \{0,1\}^{\ast}.
$$
The rule $R_v$ determines how to process a packet arriving at the node. For $R_v(\mathit{in}, F_v, s, t, b) = (\mathit{out}, b')$, the given arguments and returned values mean:
\begin{itemize}
  \item $in$ : The incoming arc from which the packet arrived (or $\varepsilon$ if the packet is originated at the node).
  \item $F_v$ : The set of locally faulty outgoing arcs.
        Note that if the node $v$ sends a message through a faulty arc, it finds the arc failed. Then that arc is added to $F_v$.
  \item $s$ : The ID of the source node of the packet. This information is stored in the packet header.
  \item $t$ : The ID of the destination node of the packet. This information is stored in the packet header.
  \item $b$ : A rewritable bit string carried in the packet header.
  \item $out$ : The outgoing arc to which the packet is sent.
  \item $b'$ : The updated bit string carried in the header of the packet sent out.
\end{itemize}
Note that the IDs of $s$ and $t$ in the packet header are never rewritten.
When a packet arrives at a node other than its destination $t$, the node decides the neighboring node to which that packet is forwarded following the rule $R_v$. After that, it sends the packet to the next node. When a node fails to send a packet, it stores the faulty arc in $F_v$ and rewrites the packet header. Then it retries the packet forwarding according to the forwarding rules.
We model the directed local failover routing as a two-player game between player 1 (the designer) and player 2 (the adversary):
\begin{enumerate}
  \item The players are given a simple directed graph $G = (V, E)$ with $n = |V|$ nodes and a failure parameter $k$.
  \item Player 1 defines a set of local forwarding rules for each node, which must be able to route packets from a given source $s$ to a destination $t$ as long as an $s$--$t$ path exists in the current graph.
  \item Let $F$ be a subset of $E$ that contains at most $k$ arcs. Player 2 removes $F$ from $G$.
  \item For any $s$ and $t$ such that every node reachable from $s$ in $G \setminus F$ is reachable to $t$ in $G\setminus F$, verify if the pre-defined local forwarding rules can successfully guide the packet from $s$ to $t$ or not.
        If it succeeds for all those pairs, then player 1 wins. Otherwise, player 2 wins.
\end{enumerate}
A local failover routing scheme $\mathcal{A}$ is an algorithm that outputs a local failover routing for a given network $G$ and a threshold parameter $k$.
The efficiency of a local forwarding rule is measured by the number of rewritable bits used in designed forwarding rules (referred to as header size).
When the header size of any packet that can be sent in routing $R(G,k)$ is less than or equal to $B$, we say that the header size of $R(G, k)$ is at most $B$. For any $n$-node directed graph $G$ and the failure parameter $k$, if the header size of $R(G,k)$ outputted by the scheme $\mathcal{A}$ is bounded by a function $f(n,k)$, the header size of the scheme $\mathcal{A}$ is said to be $f(n,k)$.

\subsection{1-Resilient Scheme $\mathcal{A}_1$ by van den Akker et al.~\cite{VWF26}}\label{sec:one-failure}
In this section, we introduce a 1-resilient local failover routing scheme $\mathcal{A}_1$ with a single rewritable bit by van den Akker et al.~\cite{VWF26}.

Without loss of generality, assume every vertex $v \in V$ is reachable from $s$ and can reach $t$; vertices satisfying neither condition are irrelevant to $s$--$t$ routing and may be removed.
Under this assumption, every arc participates in at least one $s$--$t$ path.
We call an arc $(x, y) \in E$ a \emph{non-failable arc} if removing it makes $t$ unreachable from $x$, i.e., $G \setminus (x, y)$ has no directed path from $x$ to $t$; otherwise, we call $(x, y)$ a \emph{failable arc}.
The adversary cannot select a non-failable arc as the failure without disconnecting $s$ from $t$, which violates the reachability assumption.

\paragraph{Auxiliary graph and arc-disjoint arborescences.}
Construct a multigraph $G' = (V, E')$ by initialising $E' \gets E$ and inserting one additional copy of every non-failable arc.
We seek two arc-disjoint in-arborescences rooted at $t$ in $G'$, guaranteed by the following classical result.

\begin{theorem}[Edmonds~\cite{EDM73}]\label{thm:edmonds}
  For any directed graph $H = (V', E')$ and vertex $r \in V'$, the maximum number of arc-disjoint spanning in-arborescences rooted at $r$ equals
  $$\min_{\emptyset \neq S \subseteq V' \setminus \{r\}} |\delta^{+}(S)|.$$
\end{theorem}

The following lemma verifies that $G'$ satisfies this condition with $r = t$.

\begin{lemma}\label{lem:min-cut-two}
  In $G'$, $\displaystyle\min_{\emptyset \neq S \subseteq V \setminus \{t\}} |\delta^{+}(S)| \ge 2$.
\end{lemma}
\begin{proof}
  Fix any non-empty $S \subseteq V \setminus \{t\}$.
  Since every vertex in $S$ can reach $t$, we have $|\delta^{+}(S)| \ge 1$.
  Suppose for contradiction that $|\delta^{+}(S)| = 1$, and let $(x, y)$ (with $x \in S$, $y \notin S$) be the unique arc in $\delta^{+}(S)$.
  Then $(x, y)$ is also the unique arc from $S$ to $V \setminus S$ in $G$, so its removal disconnects $x$ from $t$; hence $(x, y)$ is a non-failable arc.
  By construction, $G'$ contains two copies of $(x, y)$, so $|\delta^{+}(S)| \ge 2$, a contradiction. \qed
\end{proof}

By Lemma~\ref{lem:min-cut-two} and Theorem~\ref{thm:edmonds}, $G'$ admits two arc-disjoint in-arborescences $T_0$, $T_1$ rooted at $t$.
These trees may share non-failable arcs in $G$ (duplicate copies in $G'$ can be assigned to different trees), but this is harmless since non-failable arcs are never failed.

\paragraph{Scheme $\mathcal{A}_1$.}

Given $G$, $s$, and $t$, we first remove every vertex not reachable from $s$ or unable to reach $t$, then compute $G'$ and extract two arc-disjoint in-arborescences $T_0$, $T_1$ rooted at $t$ as described above.
For each vertex $v \ne t$, we record the outgoing arc of $v$ in $T_0$ as $\mathrm{next}_0(v)$ and the outgoing arc of $v$ in $T_1$ as $\mathrm{next}_1(v)$, yielding routing tables $\mathrm{next}_0, \mathrm{next}_1 : V \setminus \{t\} \to E$.

Each packet carries a single rewritable bit $b$, initialised to $0$ at source $s$.
At each intermediate vertex $v \ne t$ with local failure view $F_v \subseteq E$, the forwarding rule is as follows: if $b = 0$ and $\mathrm{next}_0(v) \notin F_v$, the packet is forwarded along $\mathrm{next}_0(v)$ with $b$ left at $0$; otherwise (i.e., $b = 1$, or $b = 0$ but $\mathrm{next}_0(v) \in F_v$), the packet is forwarded along $\mathrm{next}_1(v)$ and $b$ is set to $1$.

\begin{theorem}\label{thm:one-failure}
  The scheme $\mathcal{A}_1$ is a $1$-resilient local failover routing scheme using at most $1$ rewritable bit in the packet header.
\end{theorem}
\begin{proof}
  At most one arc $(x, y)$ fails.
  Since $T_0$ and $T_1$ are arc-disjoint in $G'$, and any arc shared between them in $G$ is a non-failable arc that cannot fail, the failed arc belongs to at most one of $T_0$, $T_1$ in $G$.

  \textbf{Case 1:} $(x, y) \notin T_0$.
  The packet follows $T_0$ from $s$ to $t$ with $b = 0$ throughout, reaching $t$ without interruption.

  \textbf{Case 2:} $(x, y) \in T_0$.
  The packet follows $T_0$ until reaching $x$, detects the failure of $(x, y)$, sets $b \gets 1$, and continues along $T_1$ from $x$ to $t$.
  Since $(x, y) \notin T_1$ and no further arc fails, the packet reaches $t$ successfully.

  In both cases, the packet is delivered to $t$ in at most $2n-2$ hops. \qed
\end{proof}

\subsection{$k$-Resilient Scheme $\mathcal{B}_k$ by Kawashima et al.~\cite{KKI26}}\label{sec:kki-multi}
Here, we introduce the multi-failure scheme $\mathcal{B}_k$ of Kawashima et al.~\cite{KKI26} needed for Section~\ref{sec:multi-failure}.
The scheme repeatedly attempts to forward the packet toward $t$ along arcs in $G$. The outcome of each distinct arc attempted so far is recorded in a \emph{failure profile} $c$, represented as a bit string with one bit per arc. If an attempted arc is faulty, the scheme tries another arc.

\paragraph{Scheme $\mathcal{B}_k$.}
The scheme $\mathcal{B}_k$ repeatedly attempts to route the packet from its current vertex toward $t$. Each attempted arc either succeeds, in which case the packet advances to its head, or is detected as faulty at its tail, which then tries another arc. Because all forwarding decisions are predetermined using complete knowledge of $G$, the header need only enable each vertex to reconstruct which arcs have already been attempted and the outcome of each attempt. The vertex can then deterministically recompute the next arc to try. 
More precisely, let $e_1, e_2, \ldots$ denote the sequence of attempted arcs, including both successful and failed attempts. An arc may appear more than once if the routing process returns to its tail. Let $\tau$ be the subsequence consisting of the first occurrence of each distinct arc. For each entry of $\tau$, 
the header stores one bit indicating the outcome of the corresponding attempt: $1$ if the arc was faulty and $0$ if it succeeded. We call the resulting bit string the \emph{failure profile} $c$.
Using $c$, together with $G$ and the predetermined forwarding rules, a vertex $v$ can reconstruct $\tau$ and hence the set $F'$ of arcs already found to be faulty. Because the set of faulty arcs remains fixed throughout the routing of a packet, an arc that has already succeeded cannot fail later. Accordingly, $v$ chooses, within $G \setminus F'$, an arc that reuses as many previously successful arcs as possible, thereby minimising the number of new entries added to $\tau$.

The following lemma bounds the length of $\tau$, and consequently the length of the failure profile $c$.

\begin{lemma}[\cite{KKI26}]\label{lem:kki-attempts}
  For any $k$, when $\mathcal{B}_k$ tolerates at most $k$ arc failures, at most $2n+k-3$ distinct arcs are ever attempted while routing a packet from $s$ to $t$, i.e., $|\tau| \le 2n+k-3$.
\end{lemma}
\begin{proof}[Sketch]
  Every attempted arc is of one of three types: (i) it advances the packet to a vertex not yet visited, which can happen at most $n-1$ times; (ii) it reroutes the packet to an already-visited vertex, closing a directed cycle consisting of previously-succeeded arcs, and contracting this cycle strictly decreases the number of remaining vertices, so this can happen at most $n-2$ times; or (iii) it is found faulty, which can happen at most $k$ times. Summing the three bounds gives $2n+k-3$.
\end{proof}

\begin{theorem}[\cite{KKI26}]\label{thm:kki-multi}
  For any $k$, the scheme $\mathcal{B}_k$ is a $k$-resilient local failover routing scheme using at most $\left\lceil \log \binom{2n+k-3}{k} \right\rceil$ rewritable bits in the packet header.
\end{theorem}
\begin{proof}[Sketch]
  By Lemma~\ref{lem:kki-attempts}, the failure profile $c$ has length at most $2n+k-3$ and, since at most $k$ arcs ever fail, contains at most $k$ ones. As the packet is delivered once it has encountered its last permitted failure, any trailing zero entries of $c$ may be padded to one without affecting correctness, so it suffices to distinguish failure profiles containing exactly $k$ ones among $2n+k-3$ positions; there are $\binom{2n+k-3}{k}$ such profiles, encodable in $\left\lceil \log \binom{2n+k-3}{k} \right\rceil$ bits.
\end{proof}

\section{Our Algorithms}
\subsection{An Improved Scheme for Multi Failures}\label{sec:multi-failure}

In this section, we present a local failover routing scheme $\mathcal{C}_k$ tolerating at most $k$ arc failures that improves upon the multi-failure scheme $\mathcal{B}_{k}$ of~\cite{KKI26} by using the multi-failure scheme $\mathcal{B}_{k-1}$ together with the single failure scheme $\mathcal{A}_1$ of~\cite{VWF26}.
By Theorem~\ref{thm:kki-multi}, the scheme $\mathcal{B}_k$ encodes $k$ failures among at most $2n+k-3$ arc attempts, requiring $\lceil \log \binom{2n+k-3}{k} \rceil$ bits.
By Lemma~\ref{lem:kki-attempts} applied with failure parameter $k-1$, $\mathcal{B}_{k-1}$ only ever needs to consider at most $2n+k-4$ arc attempts, recording which $k-1$ of them have failed.
Handling the remaining, $k$-th failure separately via the $1$-bit mechanism of $\mathcal{A}_1$ then yields an improvement of roughly $\log n$ bits over $\mathcal{B}_k$.

\paragraph{Scheme $\mathcal{C}_k$.}
The header carries a failure-profile string $c$ (initialised to $\mathbf{0}$) recording, among the at most $2n+k-4$ arcs that $\mathcal{B}_{k-1}$ may attempt, which ones have already failed, together with a bit $b_0 \in \{0,1\}$ (initialised to $0$).
For each of the $\binom{2n+k-4}{k-1}$ possible failure profiles $c$ with exactly $k-1$ ones, let $G_c$ denote the graph $G$ with the $k-1$ arcs recorded in $c$ removed, and, as in the construction of $\mathcal{A}_1$~(Section~\ref{sec:one-failure}), extract in-arborescences $T_0^c, T_1^c$ rooted at $t$ in $G_c$ that are arc-disjoint except on arcs that player 2 cannot remove.
At each vertex $v \ne t$, the forwarding rule proceeds in two phases: while $c$ has fewer than $k-1$ ones (Phase~1), the forwarding rule of $\mathcal{B}_{k-1}$ is applied with failure profile $c$; once $c$ reaches $k-1$ ones (Phase~2), the packet is forwarded along the arc of $T_0^c$ leaving $v$, and if that arc has failed (the $k$-th failure), $b_0$ is set to $1$ and the packet is forwarded along $T_1^c$ instead.

\begin{theorem}\label{thm:multi-failure}
  The scheme $\mathcal{C}_k$ is a $k$-resilient local failover routing scheme using at most $1+\left\lceil\log\!\binom{2n+k-4}{k-1}\right\rceil$ rewritable bits in the packet header.
\end{theorem}
\begin{proof}
  The basic proof follows the same argument as in~\cite{VWF26,KKI26}.
  The bound follows from the fact that the number of states that need to be remembered by Phase~1 and Phase~2, respectively, is $\binom{2n+k-4}{k-1}$ and $1$: Phase~1 must identify which of the $\binom{2n+k-4}{k-1}$ possible failure profiles $c$ has occurred, while Phase~2, once $c$ is fixed, needs only the single bit $b_0$ of $\mathcal{A}_1$ to choose between $T_0^c$ and $T_1^c$. \qed
\end{proof}

\paragraph{Improvement.}
$\mathcal{C}_k$ improves upon the $\lceil\log\binom{2n+k-3}{k}\rceil$ bits required by $\mathcal{B}_k$ for all $k < 2n-3$.
By Pascal's rule, $\binom{2n+k-3}{k}=\binom{2n+k-4}{k}+\binom{2n+k-4}{k-1}$.
For $k < 2n-3$, we have $\binom{2n+k-4}{k}/\binom{2n+k-4}{k-1}=(2n-3)/k>1$, so $\binom{2n+k-3}{k} \;>\; 2\binom{2n+k-4}{k-1}$.

\subsection{More Efficient 2-Resilient Scheme}

We now show that, for $k=2$, restricting $\mathcal{B}_1$ in $\mathcal{C}_2$ to the arcs of a suitably chosen $s$--$t$ path with good properties is enough to preserve $2$-resiliency, yielding a scheme $\mathcal{D}$ that further improves $\mathcal{C}_2$'s bound to $\left\lceil \log\!\left(1+2\lfloor 2(n-1)/3 \rfloor\right)\right\rceil$ bits.

As in Section~\ref{sec:one-failure}, extract two in-arborescences $T_0, T_1$ rooted at $t$ in $G$
that are arc-disjoint except for non-failable arcs.
Let $G^\ast$ be the multigraph on $V$ consisting of $T_0$ and $T_1$'s arcs kept as separate copies.
Weight each arc of $G^\ast$ by $0$ if it is one of such a shared (non-failable) pair and by $1$ otherwise (failable).
We compute a weighted shortest $s$--$t$ path $P = (u_0, u_1, \ldots, u_h)$ in $G^\ast$.
Then, the following lemma holds for $P$.

\begin{lemma}\label{lem:two-tree-shortcut}
  The number of failable arcs of $P$ is at most $\lfloor 2(n-1)/3 \rfloor$.
\end{lemma}

\begin{proof}
  Let $M = \max_{v \in V} \dist_{G^\ast}(v, t)$ and $L_i = \{v \in V \mid \dist_{G^\ast}(v, t) = i\}$ for $0 \le i \le M$.
  First, we show $|L_i| + |L_{i+1}| \ge 3$ for every $0 \le i < M$.

  Assuming $|L_j| + |L_{j+1}| < 3$ for some $j$, $|L_j| = |L_{j+1}| = 1$ holds because $L_i \neq \emptyset$ for each $0 \le i \le M$.
  Let $L_j = \{y\}$, $L_{j+1} = \{x\}$ and $S = L_{j+1} \cup \cdots \cup L_M = \{v\in V \mid \dist_{G^\ast}(v, t) \ge j+1\}$ be a nonempty subset of $V\setminus\{t\}$.
  By minimality of the distances $\dist_{G^\ast}(\cdot, t)$ defining $L_0, \ldots, L_M$, the only arcs that can leave $S$ are weight-$1$ arcs from $L_{j+1}$ to $L_j$.
  The only such arc is $(x,y)$, which, having weight $1$, is failable, hence not one of the duplicated non-failable pairs; it therefore contributes only $1$ to $\delta^+(S)$. Thus $|\delta^+(S)| \le 1$, contradicting Theorem~\ref{thm:edmonds}.

  Pairing $(L_0,L_1), (L_2,L_3), \ldots$ and summing the inequalities $|L_i|+|L_{i+1}|\ge3$ gives $n \ge 3\lfloor M/2 \rfloor + 1$ if $M$ is even (with $L_M$ left over) and $n \ge 3(M+1)/2$ if $M$ is odd (all layers paired); either way $M \le \lfloor 2(n-1)/3 \rfloor$. Since the number of failable arcs of $P$ equals $\dist_{G^\ast}(s, t) \le M$, the lemma follows. \qed
\end{proof}

\paragraph{Scheme $\mathcal{D}$.}
Let $e_1, e_2, \dots, e_{h'}$ be the failable arcs on $P$, and let $x_i$ be the tail of $e_i$.
For each $1 \le i \le h'$, apply the construction of $\mathcal{A}_1$
(Section~\ref{sec:one-failure}) to the graph $G \setminus e_i$ with source $x_i$ and destination $t$, obtaining two in-arborescences $T_0^{(i)}, T_1^{(i)}$ rooted at $t$ that are arc-disjoint except on non-failable arcs in $G \setminus e_i$.

The packet header carries a value $c \in \{0, 1, \ldots, h'\}$, initialised to $0$,
together with a bit $b_0 \in \{0,1\}$ that is only meaningful once $c \neq 0$. While
$c = 0$, the packet follows $P$, traversing shared (non-failable) arcs without changing the header state; when it reaches the tail $x_i$ of a failable arc $e_i$ and finds $e_i$ faulty, $c$ is set to $i$ and the packet is forwarded according to
$T_0^{(i)}$. Once $c = i \neq 0$, the packet follows $T_0^{(i)}$ or $T_1^{(i)}$
exactly as in the forwarding rule of $\mathcal{A}_1$.

\begin{theorem}\label{thm:2-resilient-path}
  The scheme $\mathcal{D}$ is a $2$-resilient local failover routing scheme using at most
  $\left\lceil \log\!\left(1+2\lfloor2(n-1)/3\rfloor\right)\right\rceil$
  rewritable bits in the packet header.
\end{theorem}
\begin{proof}
  Let $F \subseteq E$, $|F| \le 2$, be the set of failed arcs, and suppose every vertex reachable from $s$ in $G \setminus F$ can reach $t$ in $G \setminus F$.

  If $F \cap P = \emptyset$, the packet follows $P$ to $t$ with $c=0$ throughout (every shared arc of $P$ is healthy). Otherwise, let $e_i$ be the arc of $F \cap P$ closest to $s$ along $P$. The packet follows $P$ up to $x_i$ using only healthy arcs, detects $e_i$ faulty, and sets $c \gets i$. Since $x_i$ is reached without using any failed arc, it is reachable from $s$ in $G \setminus F$, hence, by assumption, can reach $t$ in $G \setminus F$; as $G \setminus F \subseteq G \setminus e_i$, $x_i$ can also reach $t$ in $G \setminus e_i$. The remaining failed arcs $F \setminus \{e_i\}$ number at most $1$ and all lie in $G \setminus e_i$, so, by the same case analysis as in the proof of Theorem~\ref{thm:one-failure} applied to $G \setminus e_i$, $x_i$, and $T_0^{(i)}, T_1^{(i)}$, the packet is delivered from $x_i$ to $t$ despite this single additional failure.

  By Lemma~\ref{lem:two-tree-shortcut}, $h' \le \lfloor 2(n-1)/3 \rfloor$, so the header takes at most $1+2h' \le 1+2\lfloor 2(n-1)/3\rfloor$ distinct values ($c=0$, or $c=i \in \{1,\ldots,h'\}$ together with $b_0 \in \{0,1\}$), encodable in at most $\left\lceil \log\!\left(1+2\lfloor2(n-1)/3\rfloor\right)\right\rceil$ rewritable bits. \qed
\end{proof}

\paragraph{Improvement.}
By the $k=2$ case of Theorem~\ref{thm:multi-failure}, $\mathcal{C}_2$ requires at most $1+\left\lceil \log (2n-2) \right\rceil$ bits, while $\mathcal{D}$ requires at most $\left\lceil \log\!\left(1+2\lfloor2(n-1)/3\rfloor\right)\right\rceil$ bits.
For every $n \ge 3$, $\mathcal{D}$'s bound is strictly smaller than $\mathcal{C}_2$'s, saving at least $1$ bit, on top of the improvement of $\mathcal{C}_2$ over the $\left\lceil \log \binom{2n-1}{2} \right\rceil$ bits required by $\mathcal{B}_2$ established above.
As we show in Section~\ref{sec:lower-bound} (Theorem~\ref{thm:lb-general-k}), any local failover routing scheme tolerating $2$ arc failures requires at least $\left\lceil \log \lfloor n/3 \rfloor \right\rceil$ rewritable bits: while $\mathcal{B}_2$ leaves a gap to this lower bound that grows without bound as $\Theta(\log n)$, $\mathcal{D}$ shrinks this gap to 3 bits.

\section{Improved Lower Bound for Multi Failures at Least $2$}\label{sec:lower-bound}
In this section, we show the lower bounds on the number of rewritable bits in the packet for $k \ge 2$.
To prove this, we construct a graph in which packet headers must encode a large number of failure patterns that must be distinguished for the correct delivery of packets.
First, we show the construction of a gadget $A_\ell$ as follows:
\begin{itemize}
  \item Create two paths $(u_1, u_2, \dots, u_{2\ell-1})$ and $(v_1, v_2, \dots, v_\ell)$,
  \item Add an arc $(u_{2i-1}, v_i)$ for each $1 \le i \le \ell$.
\end{itemize}
An example of $A_\ell$ is illustrated in Figure~\ref{fig:gadget}.
Here, we show the following lemma.

\begin{figure}[t]
  \centering
  \begin{subfigure}{0.49\linewidth}
    \centering
    \includegraphics[width=\linewidth]{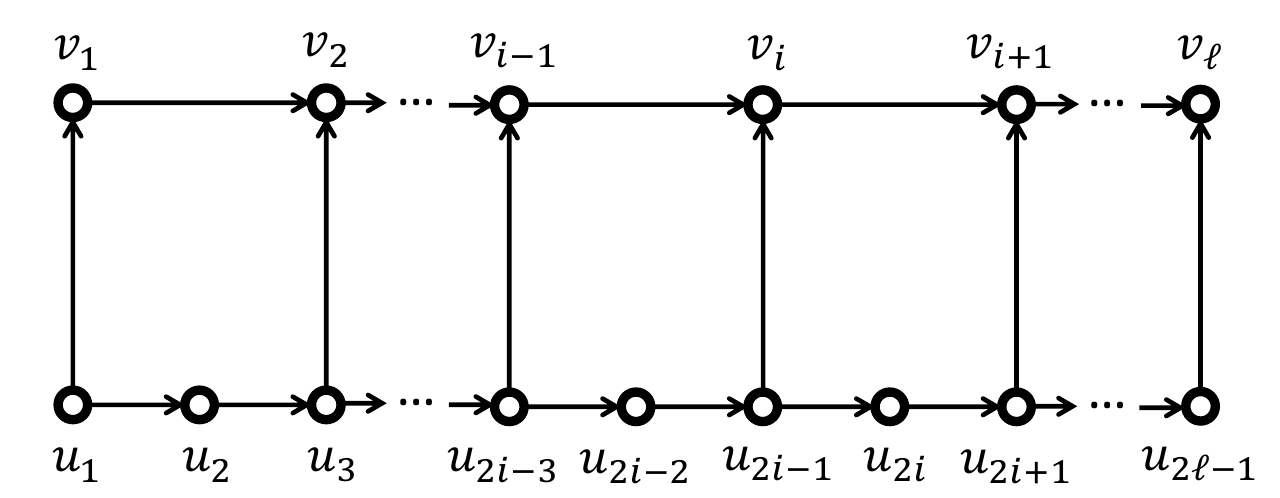}
    \caption{The example of the gadget $A_\ell$.}
    \label{fig:gadget}
  \end{subfigure}
  \hfill
  \begin{subfigure}{0.49\linewidth}
    \centering
    \includegraphics[width=\linewidth]{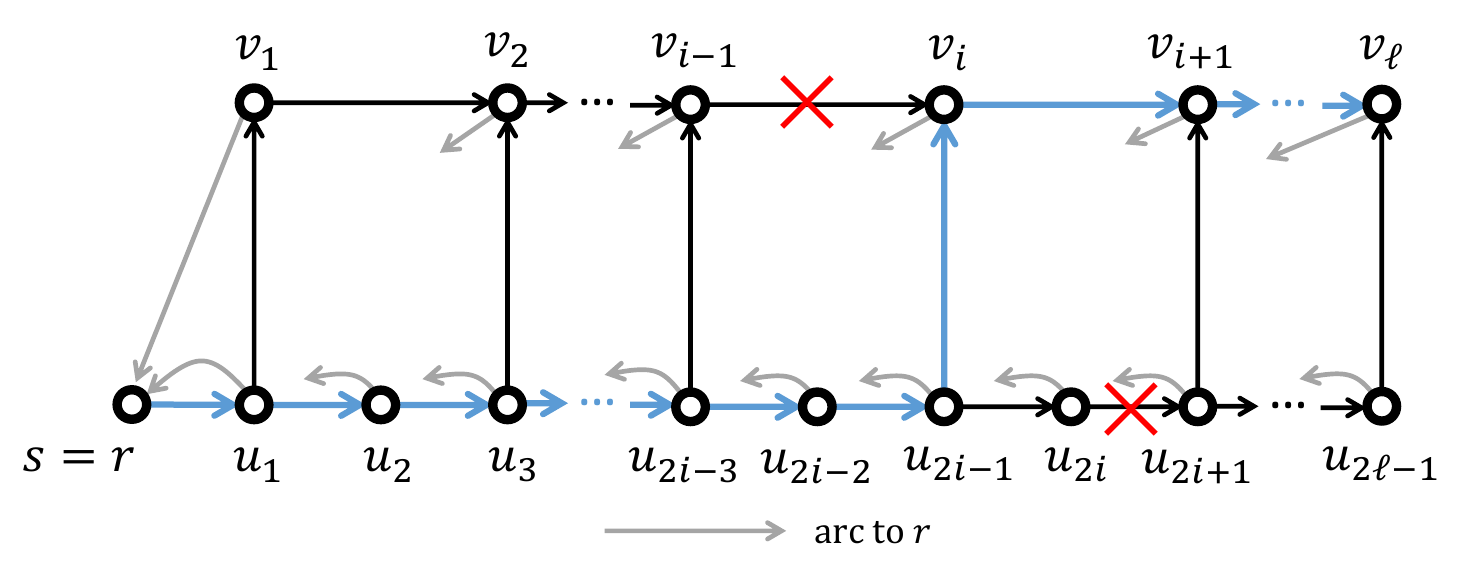}
    \caption{The strategy of player 2}
    \label{fig:lb-graph-proof}
  \end{subfigure}
  \caption{The gadget construction and the failure pattern used in the proof.}
  \label{fig:lb-graphs}
\end{figure}

\begin{lemma}\label{lem:lb-gadget}
  Prepare a vertex $r$ and a gadget $A_\ell$.
  Let $s = r$ and $t = v_\ell$, and add arcs $(u_i, r), (v_j, r)$ for $1 \le i \le 2\ell-1, 1 \le j \le \ell$ and $(r, u_1)$.
  An example of player 2's strategy is shown in Figure~\ref{fig:lb-graph-proof}.
  Let $k =2$.
  Now player 2 can decide on a set of faulty arcs so that any surviving $s$--$t$ path necessarily contains $(u_{2j-1}, v_j)$ for any $1 \le j \le \ell$.
  Precisely, if player 2 wants to enforce $s$--$t$ paths containing $(u_{2j-1}, v_j)$, it breaks the following arcs:
  \begin{itemize}
    \item $(u_{2j}, u_{2j+1})$ if $1 \le j \le \ell-1$,
    \item $(v_{j-1}, v_j)$ if $2 \le j \le \ell$.
  \end{itemize}
  In this case, at least $\lceil \log \ell \rceil$ bits are required as rewritable bits in the packet header.
\end{lemma}

\begin{proof}
  As long as player 2 follows the strategy described above, since every node has an arc to $s$ and there is a path from $s$ to $t$, any node reachable from $s$ can still reach $t$ even after removing the arcs.
  Removing $(u_{2j}, u_{2j+1})$ disconnects the path from $s$ to $t$ through $(u_{2j'-1}, v_{j'})$ for any $j < j'$.
  Similarly, removing $(v_{j-1}, v_j)$ prevents the path from $s$ to $t$ through $(u_{2j'-1}, v_{j'})$ for any $j' < j$.
  As a result, the only valid way to reach $t$ is through the arc $(u_{2j-1}, v_j)$\footnote{In local failover routing, the next hop is determined using the set of locally faulty outgoing arcs and the packet header. Therefore, if the arc $(u_{2j-1}, u_{2j})$ itself were faulty, node $u_{2j-1}$ could detect this locally without using the packet header, and routing could be performed simply by attempting to move along the arc $(u_{2j-1}, v_j)$, even without the information in the packet header. Neither of $u_{2j-1}$'s own outgoing arcs, $(u_{2j-1}, u_{2j})$ and $(u_{2j-1}, v_j)$, is ever failed by player~2; the failed arcs $(u_{2j}, u_{2j+1})$ and $(v_{j-1}, v_j)$ lie one hop away, so this header-independent shortcut is unavailable. Thus, when the packet arrives at $u_{2j-1}$, the forwarding rule must decide whether to use the arc $(u_{2j-1}, v_j)$ without observing any local failure, and this decision must therefore be made using the packet header.}.

  Suppose for contradiction that there exists a local failover routing scheme that routes the packet using at most $\lceil \log \ell \rceil -1$ rewritable bits.
  From the bound of the rewritable bit length, in this scheme, the packets can traverse at most $\ell - 1$ distinct paths between $s$ and $t$.
  Therefore, there exists an arc $(u_{2j-1}, v_j)$ that cannot be traversed regardless of the header contents.
  If player 2 adopts a strategy such that only the path that passes through this arc can reach $t$, then packet routing between $s$ and $t$ will always fail, which is a contradiction. \qed
\end{proof}

Taking $\ell = \lfloor n/3 \rfloor$ and padding with the remaining $n-3\ell \in \{0,1,2\}$ isolated, unreachable vertices already gives the $k=2$ bound: this is exactly the $k'=1$ case of the general construction below (Theorem~\ref{thm:lb-general-k}), a single gadget $A_\ell$ with hub $r$ and no cascading, to which Lemma~\ref{lem:lb-gadget} applies directly.

Next, we construct a lower bound graph for general $k$ by cascading multiple copies of the gadget $A_\ell$ (together with its hub).
By Lemma~\ref{lem:lb-gadget}, we can limit the number of paths from $s$ to $t$ in the gadget to only one by removing at most two arcs.
In this case, $\lceil \log \ell \rceil$ bits are required to remember which $(u_{2i-1}, v_i)$ in the gadget is the path from $s$ to $t$.
Thus, the total amount of information that must be encoded in the packet header grows in proportion to the number of connected gadgets.
Now, we show the following theorem.

\begin{theorem}\label{thm:lb-general-k}
  For any integer $k~(1 \le \lfloor k/2 \rfloor \le (n-1)/8)$, there exists a lower bound graph $G$ of $n$ nodes such that any local failover routing scheme tolerating at most $k$ arc failures requires at least
  \[
    \left\lceil \left\lfloor \frac{k}{2} \right\rfloor \log \left\lfloor \frac{n-1+\lfloor k/2 \rfloor}{3 \lfloor k/2 \rfloor} \right\rfloor \right\rceil
  \]
  rewritable bits in the packet header.
\end{theorem}

\begin{proof}
  Let $k' = \lfloor k/2 \rfloor, N = \lfloor (n-1+k')/3k' \rfloor$.
  The lower bound graph $G$ consists of a node $r$ and $k'$ gadgets $A_N$ as $A_N^1, A_N^2, \dots, A_N^{k'}$.
  Here, we denote any node $u_i$ in $A_N^j$ by $u_i^j$.
  Similarly, we denote any node $v_i$ in $A_N^j$ by $v_i^j$.
  Note that in our construction, we require that the number of nodes in each gadget is greater than or equal to 8.
  Hence, our construction only holds for the case $k' \le (n-1)/8$.
  We add arcs as follows:
  \begin{itemize}
    \item Arcs $(u_i^j, r)$ for each $1 \le i \le 2N-1$ and $1 \le j \le k'$
    \item Arcs $(v_i^j, r)$ for each $1 \le i \le N$ and $1 \le j \le k'$
    \item Arcs $(v_N^i, u_1^{i+1})$ for each $1 \le i < k'$
    \item An arc $(r, u_1^1)$
  \end{itemize}
  An example of $G$ is illustrated in Figure~\ref{fig:lb-general-k}.
  Let $s = r$ and $t = v_N^{k'}$.
  For each gadget $A_N^j$, player 2 determines the index $i_j~(1 \le i_j \le N)$.
  Using a technique similar to Lemma~\ref{lem:lb-gadget}, for each gadget $A_N^j$, it removes at most two arcs such that only paths through the arc $(u_{2i_j-1}^j, v_{i_j}^j)$ can reach $t$ from $s$.
  Note that player 2 can determine the index $i_j$ of each gadget after looking at the routing strategy of player 1.
  In addition, the indices of each gadget are selected independently.
  Therefore, player 1 must design a routing strategy that tries at least $N^{k'}$ routes dependent on the packet header information.
  If the correct route is not selected, the packet is necessarily returned to $s$.
  Therefore, at the time of sending a packet from $s$ to $u_1^1$, the route of the packet determined by the header information must be correct eventually.
  To decide which route to use, at least
  \[
    \left\lceil \log\!\left(N^{k'}\right) \right\rceil = \left\lceil k' \log N \right\rceil
    = \left\lceil \left\lfloor \frac{k}{2} \right\rfloor \log \left\lfloor \frac{n-1+\lfloor k/2 \rfloor}{3 \lfloor k/2 \rfloor} \right\rfloor \right\rceil
  \]
  rewritable bits. \qed

\end{proof}

\begin{figure}[t]
  \centering
  \includegraphics[width=\linewidth]{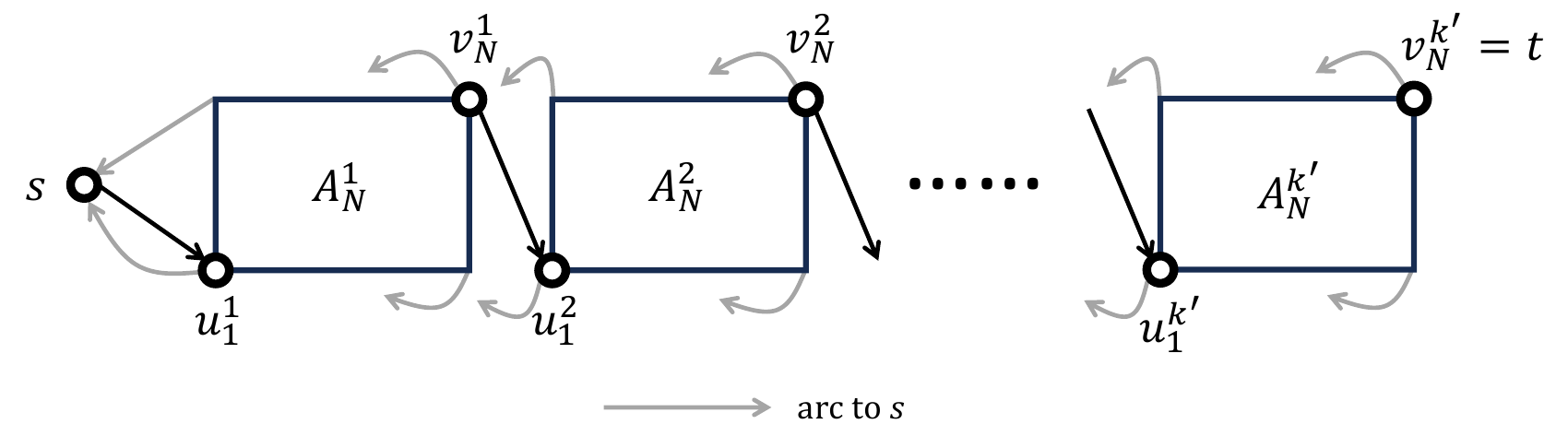}
  \caption{The example of the lower bound construction for general $k \ge 2$. The graph consists of $k'$ disjoint gadgets. Player 2 removes at most two arcs from each gadget.}
  \label{fig:lb-general-k}
\end{figure}

\section{Concluding Remarks}

We established improved upper and lower bounds on the number of
rewritable header bits required for resilient local failover routing
in directed graphs. For general $k$, our scheme $\mathcal{C}_k$
separates the handling of the final failure from that of the first
$k-1$ failures, reducing the underlying state-space bound to
$2\binom{2n+k-4}{k-1}$. For $k=2$, the path-based scheme $\mathcal{D}$
requires at most $\left\lceil \log\!\left(1+2\lfloor 2(n-1)/3 \rfloor\right)\right\rceil$ bits. 
Together with our lower bound
$\left\lceil\log\lfloor n/3\rfloor\right\rceil$, this leaves an
additive gap of at most $3$ bits.

Several questions remain open. First, determining the exact header
size for 2-resilient routing requires either a stronger lower bound or
a refinement of $\mathcal{D}$. Second, for $k>n$, the best known lower
bound is $\Omega(n)$ bits, whereas $\mathcal{C}_k$ uses
$O(n\log(k/n))$ bits. Closing this gap, particularly when $k$ is
substantially larger than $n$, is an important direction for future
work.

%
%
\bibliographystyle{alpha}
\bibliography{reference}

@misc{dai2024dynamic,
  author        = {Wenkai Dai and Klaus-Tycho Foerster and Stefan Schmid},
  title         = {On the Resilience of Fast Failover Routing Against Dynamic Link Failures},
  year          = {2024},
  eprint        = {2410.02021},
  archivePrefix = {arXiv},
  primaryClass  = {cs.NI},
  url           = {https://arxiv.org/abs/2410.02021}
}

@article{CKRRS21,
  author  = {Marco Chiesa and Andrzej Kamisi{\'n}ski and Jacek Rak and G{\'a}bor R{\'e}tv{\'a}ri and Stefan Schmid},
  title   = {A Survey of Fast-Recovery Mechanisms in Packet-Switched Networks},
  journal = {IEEE Communications Surveys \& Tutorials},
  volume  = {23},
  number  = {2},
  pages   = {1253--1301},
  year    = {2021},
  doi     = {10.1109/COMST.2021.3063980}
}

@book{RH20,
  editor    = {Jacek Rak and David Hutchison},
  title     = {Guide to Disaster-Resilient Communication Networks},
  series    = {Computer Communications and Networks},
  publisher = {Springer},
  year      = {2020},
  doi       = {10.1007/978-3-030-44685-7}
}

@inproceedings{EK24,
  author    = {Erik van den Akker and Klaus-Tycho Foerster},
  title     = {Brief Announcement: On the Feasibility of Local Failover Routing on Directed Graphs},
  booktitle = {Proc. of the 26th International Symposium on Stabilization, Safety, and Security of Distributed Systems (SSS)},
  pages     = {375--380},
  year      = {2024},
  doi       = {10.1007/978-3-031-74498-3_27}
}

@article{KKI26,
  title={Improved Algorithms for Local Failover Routing on Directed Graphs},
  author={Yuki Kawashima and Naoki Kitamura and Taisuke Izumi},
  journal={IEICE Transactions on Fundamentals of Electronics, Communications and Computer Sciences},
  volume={advpub},
  pages={2026EAP1035},
  year={2026},
  doi={10.1587/transfun.2026EAP1035}
}

@inproceedings{FGPSSS12,
  author    = {Joan Feigenbaum and Brighten Godfrey and Aurojit Panda and Michael Schapira and Scott Shenker and Ankit Singla},
  title     = {Brief announcement: on the resilience of routing tables},
  booktitle = {Proc. of the 2012 ACM Symposium on Principles of Distributed Computing (PODC)},
  pages     = {237--238},
  year      = {2012},
  doi       = {10.1145/2332432.2332478}
}

@article{VWF26,
  author = {Erik van den Akker and Marvin Weiler and Klaus-Tycho Foerster},
  title  = {Towards 1-Resilient Local Fast Failover on Directed Graphs: Flip 1 Bit Once},
  year   = {2026}
}

@inproceedings{VWF,
  author    = {Erik van den Akker and Klaus-Tycho Foerster},
  title     = {Brief Announcement: An improved lower bound for local failover routing on directed networks},
  booktitle = {Proc. of the 38th ACM Symposium on Parallelism in Algorithms and Architectures (SPAA)},
  pages     = {91--93},
  year      = {2026},
  doi       = {10.1145/3816782.3819192}
}

@article{EDM73,
  author    = {Jack Edmonds},
  editor    = {Randall Rustin},
  title     = {Edge-disjoint branchings},
  journal   = {Combinatorial Algorithms},
  publisher = {Algorithmics Press},
  pages     = {91--96},
  year      = {1973},
  url       = {https://cir.nii.ac.jp/crid/1573387451139139584}
}

@inproceedings{dai2023tight,
  author    = {Wenkai Dai and Klaus-Tycho Foerster and Stefan Schmid},
  title     = {A tight characterization of fast failover routing: Resiliency to two link failures is possible},
  booktitle = {Proc. of the 35th ACM Symposium on Parallelism in Algorithms and Architectures (SPAA)},
  pages     = {153--163},
  year      = {2023},
  doi       = {10.1145/3558481.3591080}
}

@article{chiesa2016resiliency,
  author  = {Marco Chiesa and Ilya Nikolaevskiy and Slobodan Mitrovi{\'c} and Andrei Gurtov and Aleksander M{\k{a}}dry and Michael Schapira and Scott Shenker},
  title   = {On the resiliency of static forwarding tables},
  journal = {IEEE/ACM Transactions on Networking},
  volume  = {25},
  number  = {2},
  pages   = {1133--1146},
  year    = {2017},
  doi     = {10.1109/TNET.2016.2619398}
}

@inproceedings{foerster2021feasibility,
  author    = {Klaus-Tycho Foerster and Juho Hirvonen and Yvonne-Anne Pignolet and Stefan Schmid and Gilles Tredan},
  title     = {On the Feasibility of Perfect Resilience with Local Fast Failover},
  booktitle = {Proc. of the 2nd Symposium on Algorithmic Principles of Computer Systems (APOCS)},
  pages     = {55--69},
  year      = {2021},
  doi       = {10.1137/1.9781611976489.5}
}

@inproceedings{grobe2024local,
  author    = {Jonas Grobe and Stephanie Althoff and Klaus-Tycho Foerster},
  title     = {Local Fast Failover Routing on Directed Networks},
  booktitle = {Proc. of the 2024 14th International Workshop on Resilient Networks Design and Modeling (RNDM)},
  pages     = {1--8},
  year      = {2024},
  doi       = {10.1109/RNDM64105.2024.10820439}
}

@inproceedings{van2024short,
  author    = {Erik van den Akker and Klaus-Tycho Foerster},
  title     = {Short paper: Towards 2-resilient local failover in destination-based routing},
  booktitle = {Proc. of the 9th International Symposium on Algorithmic Aspects of Cloud Computing (ALGOCLOUD)},
  pages     = {17--25},
  year      = {2024},
  doi       = {10.1007/978-3-031-94677-6_2}
}
\end{document}